\documentclass[10pt]{amsart}

\usepackage{amsfonts}
\usepackage{amssymb}
\usepackage{amsmath,amsthm}
\usepackage{amscd}
\usepackage{enumitem}
\usepackage{cancel}
\usepackage{pdfsync}
\usepackage{mathrsfs}
\usepackage{color}
\usepackage{multicol}
\usepackage{booktabs}
\usepackage{psfrag,graphicx,float}
\usepackage[mathscr]{eucal}
\usepackage{stmaryrd}
\usepackage[colorlinks,linkcolor=blue,citecolor=blue,urlcolor=blue]{hyperref}

\numberwithin{equation}{section}

\newcommand{\beq}{\begin{equation}}
\newcommand{\eeq}{\end{equation}}
\newcommand{\bea}{\begin{eqnarray}}
\newcommand{\eea}{\end{eqnarray}}
\newcommand{\nn}{\nonumber}
\newcommand\noi{\noindent}

\newcommand{\bk}{\begin{cases}}
\newcommand{\ek}{\end{cases}}
\newcommand{\tbf}{\textbf}

\newtheorem{definition}{Definition}

\newtheorem{theorem}{Theorem}
\newtheorem{corollary}{Corollary}
\newtheorem{lemma}{Lemma}
\theoremstyle{definition}

\newtheorem{remark}{\textbf{Remark}}

\usepackage{tikz}
\usetikzlibrary{positioning}

\tikzset{>=stealth}

\begin{document}

\author{Daniel Reyes}
\address{Departamento de F\'{\i}sica Te\'{o}rica, Facultad de Ciencias F\'{\i}sicas, Universidad
Complutense de Madrid, 28040 -- Madrid, Spain 
\\ and Instituto de Ciencias Matem\'aticas, C/ Nicol\'as Cabrera, No 13--15, 28049 Madrid, Spain
}
\email{danreyes@fis.ucm.es}
\author{Miguel A. Rodr\'iguez}
\address{Departamento de F\'{\i}sica Te\'{o}rica, Facultad de Ciencias F\'{\i}sicas,
  Universidad Complutense de Madrid, 28040 -- Madrid, Spain }
\email{rodrigue@ucm.es}
\author{Piergiulio Tempesta}
\address{Departamento de F\'{\i}sica Te\'{o}rica, Facultad de Ciencias F\'{\i}sicas,
  Universidad Complutense de Madrid, 28040 -- Madrid, Spain \\ and Instituto de Ciencias
  Matem\'aticas, C/ Nicol\'as Cabrera, No 13--15, 28049 Madrid, Spain}
\email{p.tempesta@fis.ucm.es, piergiulio.tempesta@icmat.es}

\title[A Frobenius-type theory for discrete systems]{A Frobenius-type theory for discrete systems}


\date{July 16th, 2024}

\begin{abstract}
We develop an approach analogous to classical Frobenius theory for the analysis of singularities of ODEs in the case of discrete dynamical systems. Our methodology is based on the Roman-Rota theory of finite operators and relies crucially on the idea of preserving the Leibniz rule on a lattice of points by means of the notion of Rota algebras. The relevant cases of the Bessel, Hermite and Airy equations are discussed. 
\end{abstract}

\maketitle

\tableofcontents

\section{Introduction}

One of the most relevant problems in modern mathematical physics is the study of discretization approaches that can preserve 
the symmetry and integrability properties of dynamical models expressed in terms of ODEs and PDEs (see, e.g., \cite{BS, LNP, LWY, Suris}). Discrete formulations of quantum theories have been proposed, for instance, in several contexts of quantum gravity (see, e.g., \cite{FL, GP, RS}).

In the last decades, the \textit{finite operator theory} (also called Umbral Calculus), developed by G. C. Rota and his collaborators \cite{Roman, RR, Rota} as a unifying language for the study of sequences of polynomials and combinatorial problems, has also been used as a novel theoretical framework for discretizing quantum mechanical models  \cite{DMS, LNO, LTW1, LT2011, LTW2} and for the multi-scale analysis of dynamical systems on a lattice \cite{LT2011}.

A consistent approach, based on this philosophy, has been formulated more recently: it requires that the discrete operators act as derivations with respect to a non-standard multiplication rule; in other words, they must preserve the Leibniz rule on lattices. This idea was proposed, from different perspectives and contexts,  in \cite{BF, Ward, Ismail} and then in \cite{PTJDE, PTprep}.

From a technical point of view, the latter approach requires the definition of a suitable non-local product of functions, which acts on the basic polynomials associated with a given discrete derivative (delta operator) in the same way as the standard local product acts on monomials. 
In this way, the whole class of exact solutions expressed in terms of power series is preserved. Formally, the algebraic structure in which the discretization process will be implemented is that of the Rota algebra, introduced in \cite{PTJDE}. It is in fact a Galois differential algebra of formal power series, equipped with a delta operator and a suitable product, which makes it act as a derivation. 

A second technical requirement, apart from preserving the Leibniz rule, is the identification of the lattice of points where the dynamics takes place with the set of zeros of the basic polynomials associated with a given delta operator. This fundamental requirement solves the main problem associated with the previous ``umbral'' approaches to the discretization problem, namely, the fact that the theory was essentially formulated in the language of formal power series, typically divergent.

In \cite{PTJDE, PTprep}, this approach has been formulated in a categorical framework. In particular, the discrete analogues of a large class of both linear and nonlinear dynamical systems, whose coefficients are expressed in terms of polynomials, have been determined. A recent related study, aimed at the discretization of the Euler equation, has been proposed in \cite{RT2024}.

In this article, we shall focus more specifically on the study of the singularity properties of discrete versions of second-order linear ODEs with \textit{analytical} coefficients. Also, from a methodological point of view, we adopt a different strategy with respect to \cite{PTJDE, PTprep}. In fact, our proofs are new and independent, since they do not require the use of category theory. 

Our main result is a new theoretical framework, formalized in terms of theorems \ref{Theorem1}-\ref{Theorem3}, which can be considered as a \textit{discrete analogue of the classical Frobenius approach} to second-order ODEs. This framework  allows us to construct a family of exact solutions for a large class of second-order discrete dynamical models, according to the type of singularities they possess. Indeed, we can discretize ODEs admitting both ordinary and singular regular points, in such a way that a class of exact solutions of the continuous case are preserved by construction. To our knowledge, there is no analogous result in the literature on discrete systems.

We will also apply our theorems to the determination of discrete formulations of classical ODEs, such as the Airy, Hermite and Bessel equations.

In our discussion, we will take a constructive approach, based on finite operator calculus, and for simplicity we will derive all results in a self-contained manner.

In \cite{Hilger1990}, an interesting approach, nowadays known as time scale calculus, based on the idea of unifying discrete and continuous calculus, was initiated. From a methodological point of view, this unification is realized by defining a general domain that can be continuous, discrete, or mixed (time scales or, more generally, measure chains) and suitable jump operators on this domain. It has been actively investigated in the last years by many authors (see, e.g., the recent monograph \cite{BG2017}), and a large body of results and applications is available. However, the theoretical framework proposed in this article is independent and essentially different from the time series calculus, being based on the alternative ideas of preserving the Leibniz rule in discrete calculus and of using, for the discretization, an adapted lattice of points.

The structure of the  paper is as follows. In Section \ref{sec:2},  we review some basic notions of finite operator theory. 
In Section \ref{sec:3}, we introduce a generalization of Frobenius theory on a regular lattice of points. In particular, we  study the case of ODEs admitting ordinary points, and prove Theorem \ref{Theorem1}. In Section \ref{sec:4}, we extend our analysis to the case of singular regular points, leading to Theorem \ref{Theorem2}. In Section \ref{sec:5}, we apply our approach to define a discrete version of the Bessel equation. A novel discretization scheme on the real negative axis, related to the backward discrete derivative, and the related Theorem \ref{Theorem3} are proposed in Section \ref{sec:6}. Some future research perspectives are outlined in the final Section \ref{sec:7}.

\section{The finite operator calculus: Background and notation}
\label{sec:2}
The purpose of this section is to briefly review some basic notions and fundamental results of the Roman-Rota's formulation of the theory of difference operators, which provides a very natural language for defining the theory of basic, Appell and Sheffer sequences of polynomials.

This approach can be regarded as the modern version of the  Umbral Calculus of the XIX century introduced by Sylvester, Cayley, Blissard, and other authors \cite{BL2000}. 

\subsection{Basic notions}
Let $\mathbb{N}$ be the set of non-negative integers, $\mathbb{K}$ be a field of characteristic zero and $\mathcal{F}[[\mathbb{K}]]$ the algebra of formal power series in $x$. We shall denote by   $\mathcal{P}$ the space of polynomials in a variable $x\in \mathbb{K}$.  Throughout this article, $\mathbb{K}=\mathbb{R}$. 
\begin{definition}
An operator $T^{h}: \mathcal{P} \to \mathcal{P}$ which satisfies the relation
\[
T^{h} p(x)= p(x+h), \qquad h>0, 
\]
will be said to be a \textit{shift operator}. 
\end{definition}
From now on, for simplicity, we will use the notation $T:=T^{h}$, for $h=1$.

\begin{definition}\label{deltaop}
An operator $S$  is said to be {\it shift-invariant} if it
commutes with all the shift operators $T^{h}$.
\end{definition}
Let $\mathcal{A}$ be the algebra of shift-invariant operators, equipped with the usual operations: the sum and the product of
two operators, and the product of a scalar with an operator.
\begin{definition}
 A shift-invariant operator $Q$ is said to be a \textit{delta operator} if $Q \, x=const\neq0$.
\end{definition}
One can immediately prove the following
\begin{corollary}\label{cor1}
For each constant $c\in \mathbb{R}$, $Q \, c=0$.
\end{corollary}
The notion of delta operator has been investigated extensively  in the literature (see, for example, \cite{Rota}). Very typical examples are  the standard derivative $D$, the forward discrete derivative $\Delta^{+}=T-\mathbf{1}$, the backward derivative $\Delta^{-}=\mathbf{1}-T^{-1}$ and the symmetric derivative $\Delta^{s}=\frac{T-T^{-1}}{2}$. 

If we denote the operator $D^k$ by $t^k$, its action on a polynomial is given by $t^k p(x)= p^{(k)}(x)=\frac{d^k p(x)}{dx^k}$. By extending this action by linearity, we can also associate a differential operator with a formal power series:
\[
F(t)=\sum_{k}\frac{a_k t^k}{k!} \Longrightarrow F(D)= \sum_{k} \frac{a_k D^k}{k!}.
\]
This correspondence allows us to represent delta operators in terms of a function $q(t)$, called the \textit{representative} of $Q$.  For instance, a non-standard example is the Gould operator $q^{G}(t)=e^{at}(e^{bt}-1)$. Another interesting case is the Abel operator $q^{A}(t)=t e^{\sigma t}$, whose basic sequence is the set of Abel polynomials $p^{A}_n(x)=x(x-n \sigma)^{n-1}$ \cite{Roman}.

\begin{definition}\label{def2}
 A polynomial sequence $\{p_{n}\left(x\right)\}_{n\in\mathbb{N}}$ is said to be the sequence of \textit{basic polynomials} associated with a delta operator $Q$ if the following conditions are satisfied:
\begin{align}
&  1)\text{ \ }p_{0}\left(  x\right)  =1;\notag\\
&  2)\text{ \ }p_{n}\left(
0\right)  =0\ \text{for all }n>0;\text{ \ }\nonumber\\
&  3)\text{ \ }Q p_{n}\left(  x\right)  =np_{n-1}\left(  x\right).\nonumber
\end{align}
\end{definition}
As is well known, there exists a unique sequence of  basic polynomials associated with a delta operator \cite{Rota}. For instance, for $Q= D$,  the sequence of basic polynomials is given by $p_n(x)=x^n$. For the forward and backward discrete derivatives $\Delta^{\pm}$, we have 
\begin{equation}
p_{0}^{\pm}(x)=1, \hspace{2mm} p_n^{\pm}(x):=x(x\mp 1)(x \mp 2 )...(x\mp (n-1)), \label{eq:2.1}
\end{equation}
respectively. For $\Delta^{s}$, the basic polynomials are expressed in terms of Gould polynomials \cite{LT2011}, which are the basic sequence for the operator $q^{G}$. 

\subsection{Pincherle derivative and basic polynomials}
We revise a technique, discussed in \cite{LTW2}, which allows us to determine the basic polynomials associated with an arbitrary delta operator by using the algebraic notion of the Pincherle derivative.  Let $x: p(x) \to x p(x)$ denote the multiplication
operator. 
\begin{definition}
The \textit{Pincherle derivative} of a delta operator $Q\in \mathcal{A}$ is defined by the relation
\begin{equation}\label{PD}
Q':=[Q,x].
\end{equation}
\end{definition}
If $Q$ is a delta operator, then $Q'$ is invertible \cite{Rota}. We shall determine an operator $\beta\in \mathcal{A}$ such that the Heisenberg-Weyl algebra is satisfied \cite{LNO}:
\begin{equation}
[Q, x\beta]= \mathbf{1}.
\end{equation}
Given a delta operator $Q$, we will call $\beta$ the conjugate operator of $Q$. This operator exists because it is uniquely determined by the relation $\beta= (Q')^{-1}$. 
It is easy to see that for $Q=D$, $\beta= \mathbf{1}$; for $Q=\Delta^{+}$, $\beta= T^{-1}$; for $Q= \Delta^{-}$, $\beta=T$. 

Thus, the basic polynomials associated with a delta operator $Q$  can be determined by the simple relation 
\begin{equation}
p_n(x)= (x\beta)^n\cdot 1 \ .
\end{equation}

\subsection{Discretization of dynamical systems: Rota algebras}
First, we observe that for a given $Q$, its family of basic polynomials $\{p_n(x)\}_{n\in\mathbb{N}}$ provides a basis for the algebra of formal power series $\mathcal{F}$. Thus, any $f\in\mathcal{F}$ can be expanded into a formal series of the form 
\[
f(x)=\sum_{n=0}^{\infty}a_n p_n(x) \label{exp}.
\]
We also introduce in $\mathcal{F}$ the  $*$ product, which is associative and commutative, defined by the relation
\begin{equation}
p_n(x)*p_m(x):=p_{n+m}(x). \label{starproduct}
\end{equation}
This product for the forward difference operator $\Delta$ was defined in \cite{Ward} and \cite{Ismail}. Also, the $*$ product can naturally  be extended by linearity to  spaces of formal power series. If $f, g\in \mathcal{F}$ are defined on a set of points $\mathcal{L}\in \mathbb{R}$ (usually referred to as the lattice $\mathcal{L}$) and are expanded in the polynomial basis provided by the basic sequence $\{p_n(x)\}_{n\in \mathbb{N}}$, we have the fundamental relation
\begin{equation}\label{eq:der}
\Delta (f*g) = (\Delta f)*g + f*(\Delta g).
\end{equation}
In other words, $\Delta$ acts as a derivation with respect to the $*$-product: the Leibniz rule is restored on $\mathcal{F}$. Moreover, if we identify the set of points of our lattice $\mathcal{L}$ with the zeros of the basic polynomials, then all the formal series involved are truncated. 

Starting from these ideas, for any delta operator $Q$ one can define the notion of \textit{Rota algebra}, introduced in \cite{PTJDE} as a natural environment where one can carry out the discretization procedure described above. 
  
\begin{definition}
A Rota differential algebra is a Galois differential algebra $(\mathcal{F}, Q)$, where $(\mathcal{F}, +, \cdot, *_{Q})$ is an associative algebra of formal power series, the product $*_Q$ is the composition law defined by \eqref{starproduct}, and $Q$ is a delta operator acting as a derivation on $\mathcal{F}$:
\begin{equation}
i) \quad Q(a+b)=Q(a)+Q(b), \hspace{10mm} Q(\lambda a)=\lambda Q(a), \quad \lambda\in\mathbb{K},
\end{equation}
\begin{equation}
ii) \quad Q(a*b)= Q(a)*b+a*Q(b).
\end{equation}
\end{definition}

\noi As was proved in \cite{PTJDE}, there exists a unique Rota algebra $(\mathcal{F}, Q)$ 
associated with a delta operator $Q$.

Thus, our discretization amounts to associating with a given ODE a formal equation defined in a Rota algebra, obtained by replacing the continuous derivative $D$ by a delta operator $Q$ (acting as a Galois operator due to eq.\eqref{eq:der}), and the basic sequence $x^{n}$ by the basic polynomials $p_n(x)$.

\begin{center}
\begin{tikzpicture}
    \node (E) at (0,0) {$D$};
    \node (F) at (2,0) {$Q$};
    \node (N) at (2,-2) {$p_{n} (x)$};
    \node (M) at (0,-2) {$x^{n}$};
    \draw[->] (E)--(F) node [midway,above] {}; 
    \draw[<->] (F)--(N) node [midway,right] {}; 
    \draw[->] (M)--(N) node [midway,below] {}; 
    \draw[<->] (E)--(M) node [midway,left] {}; 
\end{tikzpicture}
\end{center}

This discretization scheme extends the  \textit{umbral correspondence} discussed in \cite{Rota,DMS,LTW1,LTW2}, because 
is no longer defined on formal power series, but is \textit{effective}, i.e., due to the truncation phenomenon described above, it generates standard finite-difference dynamical models which admit convergent solutions in one-to-one correspondence with analytic solutions of the original continuous model \cite{Ward,PTJDE,PTprep}. Different choices of $Q$ (and hence of the associated Rota algebra) lead to different discrete models, sharing the same family of solutions, \textit{mutatis mutandis}. In this technical sense, we say that this discretization preserves integrability.

\section{Generalizing Frobenius analysis on the lattice: Regular points} \label{sec:3}

\subsection{Statement of the problem}

The purpose of our analysis is to define a discretized version of the second order linear differential equation
\beq\label{cont}
u''+A(x)u'+B(x)u=0. 
\eeq
More generally, we wish to construct a theory on a discrete lattice of points that preserves the main results of the standard Frobenius approach based on the singularity analysis of the equation \eqref{cont}.

According to the standard terminology, if $A(x)$ and $B(x)$ are analytic functions in a neighbourhood of a given point $x_0$, we shall say that $x_0$ is a regular point for the equation; otherwise, $x_0$ is singular. 
Assuming that $x=0$ is a regular point, let
\beq \label{eq:3.2}
A(x)=\sum_{n=0}^{\infty}a_n x^n,\quad B(x)=\sum_{n=0}^{\infty} b_n x^n
\eeq
be the series expansions of $A(x)$ and $B(x)$ around $x=0$. As is well known \cite{Strauss}, there exists a basis of the space of solutions of eq. \eqref{cont} formed by two analytic functions with convergence radius $R\geq R_{0}:=\min \{R_1,R_2\}$, where  $R_1$ and $R_{2}$ are the convergence radii of the series representing $A(x)$ and $B(x)$. The coefficients in the series expansions of the solutions can be  uniquely determined once the initial conditions are fixed. Let
\beq \label{eq:3.3}
u(x)=\sum_{k=0}^{\infty} \zeta_k x^k
\eeq
be the general form of the two independent series solutions of eq. \eqref{cont}.
We now introduce  the discrete counterpart of eq. \eqref{cont}. Specifically,  we introduce the formal series denoted (with an abuse of notation) by
\beq \label{eq:3.4}
u(x)=\sum_{k=0}^{\infty} \zeta_kp_k(x),
\eeq
where $p_k(x)$ are the basic polynomials associated with a given delta operator $Q$. Let us also define the formal series
\beq
F(x):=\sum_{k=0}^{\infty}a_k p_k(x), \qquad G(x):= \sum_{k=0}^{\infty}b_k p_k(x), 
\eeq
where $\{a_k\}_{k\in\mathbb{N}}$ and $\{b_k\}_{k\in\mathbb{N}}$ are the coefficients of the series \eqref{eq:3.2}.

\begin{definition}
Let $Q$ be a delta operator and $\{p_n(x)\}_{n\in \mathbb{N}}$ be the family of associated basic polynomials. We introduce the class of discrete dynamical systems
\beq\label{discr}
Q^2 u(x) +F(x)* Q u(x) +G(x)*u(x)=0
\eeq
defined on the Rota algebra $(\mathcal{F}, +,\cdot, *_{Q})$.
\end{definition}

\begin{lemma} \label{Lema1}
The discrete model \eqref{discr} admits the formal solution
\beq 
u(x)=\sum_{k=0}^{\infty} \zeta_kp_k(x),
\eeq
where the coefficients $\zeta_k$ are defined by the recursion
\beq \label{eq:3.8}
\zeta_{k+2}
=-\frac{1}{(k+2)(k+1)}  \sum _{m=0}^k \left[(m+1) a_{k-m} \zeta_{m+1} 
+  b_{k-m}\zeta_{m}\right].
\eeq
\end{lemma}

\begin{proof}
Explicitly, eq. \eqref{discr} reads
\begin{align*}
&\sum_{k=0}^{\infty}k(k-1)\zeta_kp_{k-2}(x)
+\left(\sum_{k=0}^{\infty}a_kp_k(x)\right)*
\left(\sum_{k=0}^{\infty} k\zeta_kp_{k-1}(x)\right)\notag\\
&+\left(\sum_{k=0}^{\infty}b_kp_k(x)\right)*
\left(\sum_{k=0}^{\infty} \zeta_kp_k(x)\right)=0.
\end{align*}
Using eq. \eqref{starproduct}, it can be written as
\beq \nn \label{diff3}
\sum_{k=0}^{\infty}k(k-1)\zeta_kp_{k-2}(x)
+\sum_{k,l=0}^{\infty}k a_l\zeta_kp_{k+l-1}(x)
+ \sum_{k,l=0}^{\infty}b_l\zeta_kp_{k+l}(x) =0.
\eeq
After some simplification, we obtain
\beq
\sum_{k=0}^{\infty}\left[(k+2)(k+1)\zeta_{k+2}
+  \sum _{m=0}^k \left[(m+1) a_{k-m} \zeta_{m+1} 
+  b_{k-m}\zeta_{m}\right]\right]p_{k}(x) =0.
\eeq
This relation implies the recurrence \eqref{eq:3.8}.
\end{proof}

\begin{remark}
As a direct consequence of the previous Lemma, the class of analytic solutions of eq. \eqref{cont} are inherited by every discrete representative of eq. \eqref{discr}, in the sense that every solution of the form \eqref{eq:3.3} of eq. \eqref{cont} defines a solution of the form \eqref{eq:3.4} of eq. \eqref{discr}. 
\end{remark}

\subsection{Difference equations}

The family of systems \eqref{discr} can also be represented in the form of a standard difference equation,  once a difference delta operator $Q$ is chosen, for a dependent variable $u_n$ defined on a suitable lattice of points $\mathcal{L}$ 
given by the set of zeros of the basic polynomials associated with $Q$. 

In the subsequent considerations, we will focus on the case $Q=\Delta^{+}$. Thus, the lattice $\mathcal{L}^{+}$ is defined by the set of points $x=n$, $n=0,1,2,\ldots$. We introduce the function
\beq 
u_n:=\sum_{k=0}^{n} \zeta_k p_k^{+}(n)=  \sum_{k=0}^{n} \zeta_k \frac{n!}{(n-k)!}.
\eeq
This relation can be inverted \cite{PTJDE}:
\beq \label{inv}
\zeta_k= \sum_{j=0}^{k}\frac{(-1)^{k-j}}{j!(k-j)!}u_j .
\eeq
From the recurrence \eqref{diff3} we obtain:
\beq\label{eq1}
\sum_{k=0}^{\infty}\left[k(k-1) p_{k-2}^{+}(n)+ k\sum_{l=0}^{\infty} a_l p_{k+l-1}^{+}(n) +\sum_{l=0}^{\infty}b_l p_{k+l}^{+}(n)\right]\zeta_k =0,
\eeq
that is\footnote{In these finite sums we keep the upper limit equal to infinity to simplify the notation, taking into account that factorials of negative integers in the denominator yield vanishing addends.}:
\beq \label{eq:3.13}
\sum_{k=0}^{\infty}\left[\sum_{j=0}^{k}\left(\frac{n! \, k(k-1) }{(n+2-k)!}+  \sum_{l=0}^{\infty}\left[ \frac{n!\, k a_l}{(n-k-l+1)!}  + \frac{n!\, b_l }{(n-k-l)!}\right]\right)\frac{(-1)^{k-j}}{j!(k-j)!}u_j \right]=0.
\eeq
Let us introduce the coefficient $T_{kj}$, defined by
\beq\label{coeT}
T_{kj}:=\left(\frac{n!\,k(k-1) }{(n+2-k)!}+  \sum_{l=0}^{\infty}\left[ \frac{n!\,k a_l}{(n-k-l+1)!}  + \frac{n!\,b_l }{(n-k-l)!}\right]\right)\frac{(-1)^{k-j}}{j!(k-j)!}.
\eeq
Thus, eq. \eqref{eq:3.13} reads
\beq \label{eq:3.15}
\sum_{k=0}^{\infty}\sum_{j=0}^{k}T_{kj}u_j=0.
\eeq
By exchanging the indices, we get a linear difference equation for $u_j$ 
\beq \label{eq:3.16}
\sum_{j=0}^{\infty}\Bigg(\sum_{k=j}^{\infty}T_{kj}\Bigg)u_j=\sum_{j=0}^{\infty}C_ju_j=0,\quad C_j:= \sum_{k=j}^{\infty}T_{kj},
\eeq
where all sums have only a finite number of nonvanishing addends.
\subsubsection{Coefficients $C_j$}
The coefficients $C_j$ in eq. \eqref{eq:3.16} explicitly read
\beq
C_{j}=\sum_{k=j}^{\infty}\left(\frac{n!k(k-1) }{(n+2-k)!}+  \sum_{l=0}^{\infty}\left[ \frac{n!k \, a_l}{(n-k-l+1)!}  + \frac{n!\,b_l }{(n-k-l)!}\right]\right)\frac{(-1)^{k-j}}{j!(k-j)!}.
\eeq
The first term in this expression can be easily simplified\footnote{We use the following identities, which can be easily proved:
\[\sum _{k=0}^n (-1)^k \binom{n}{k}=\delta _{n,0},\quad \sum _{k=0}^n (-1)^k k\binom{n}{k}=-\delta _{n,1},\quad \sum _{k=0}^n (-1)^k k^2 \binom{n}{k}=2\delta _{n,2}-\delta _{n,1}.\]}. We have
\begin{align}
\sum_{k=j}^{n+2}\frac{n!k(k-1) }{(n+2-k)!}&\frac{(-1)^{k-j}}{j!(k-j)!}  =\frac{1}{(n+1)(n+2)}
\binom{n+2}{j} \sum_{k=j}^{n+2}(-1)^{k-j}k(k-1)\binom{n+2-j}{k-j}\notag\\ = \nn & \frac{1}{(n+1)(n+2)}
\binom{n+2}{j}[(n+1)(n+2)\delta _{n+2,j}-2(n+1) \delta_{n+1,j}+2\delta_{n,j}]. \label{first}
\end{align}
Thus, taking the sum in $j$, we obtain the relation
\beq
\sum_{k=j}^{n+2}\frac{n!k(k-1) }{(n+2-k)!}\frac{(-1)^{k-j}}{j!(k-j)!} = \delta _{n+2,j}-2 \delta_{n+1,j}+\delta_{n,j}.
\eeq

As for the second term, we have
\beq
\sum_{k=j}^{\infty}  \sum_{l=0}^{\infty} \frac{n!k \,a_l}{(n-k-l+1)!}  \frac{(-1)^{k-j}}{j!(k-j)!}=\sum_{l=0}^{\infty}a_l\sum_{k=j}^{n-l+1}   \frac{n!k}{(n-k-l+1)!}  \frac{(-1)^{k-j}}{j!(k-j)!}.
\eeq
Then, the expression
\beq
\sum_{k=j}^{n-l+1}   \frac{n!k}{(n-k-l+1)!}  \frac{(-1)^{k-j}}{j!(k-j)!}
\eeq
can be simplified as
\begin{align}
&\frac{n!}{(n-l+1)!}\binom{n-l+1}{j}\sum_{k=j}^{n-l+1}(-1)^{k-j}\binom{n-l+1-j}{k-j}k
= \frac{n!}{(n-l)!}(\delta_{n-l+1,j}-\delta_{n-l,j}).
\end{align}
Therefore, the second term can be rewritten as
\beq \label{eq:3.23}
\sum_{l=0}^{\infty}\frac{n!\, a_l}{(n-l)!}(\delta_{n-l+1,j}-\delta_{n-l,j}).
\eeq
The third term reads
\beq
\sum_{l=0}^{\infty}b_l\sum_{k=j}^{n-l} \frac{n! }{(n-k-l)!}\frac{(-1)^{k-j}}{j!(k-j)!},
\eeq
and taking into account that
\begin{align}
\sum_{k=j}^{n-l} \frac{n!}{(n-l-k)!}\frac{(-1)^{k-j}}{j!(k-j)!} & =\frac{n!}{(n-l)!}\binom{n-l}{j}\sum_{k=j}^{n-l} (-1)^{k-j}\binom{n-l-j}{j-k}\notag\\ \nn  = & \frac{n!}{(n-l)!}\binom{n-l}{j}\delta_{n-l,j} =\frac{n!}{(n-l)!}\delta_{n-l,j},
\end{align}
it assumes the simpler form
\beq \label{eq:3.26}
\sum_{l=0}^{\infty}\frac{n!\, b_l}{(n-l)!}\delta_{n-l,j}.
\eeq
Finally, collecting together the previous results, we have proved that eq. \eqref{eq:3.15} can be rewritten as
\begin{align}
 \sum_{j=0}^{\infty}C_{j}u_j=&\sum_{j=0}^{\infty}\left(\delta _{n+2,j}-2 \delta_{n+1,j}+\delta_{n,j}\right)u_j+  \sum_{l=0}^{\infty} \sum_{j=0}^{\infty}\frac{n! \, a_l}{(n-l)!}(\delta_{n-l+1,j}-\delta_{n-l,j})u_j\notag \\ & + \sum_{l=0}^{\infty}\sum_{j=0}^{\infty}\frac{n!\, b_l}{(n-l)!}\delta_{n-l,j}u_j=0. 
\end{align}
As a result of the latter analysis and of Lemma \ref{Lema1}, we have proved the following theorem.
\begin{theorem} \label{Theorem1}
Let us define the family of difference equations
\beq \label{fam+}
u_{n+2}-2 u_{n+1}+u_{n}+  \sum_{l=0}^{n}\frac{n!}{(n-l)!}\left[a_l u_{n-l+1}-(a_l- b_l)u_{n-l}\right]=0,
\eeq
where $n\in \mathbb{N}$ and $a_l$, $b_l\in \mathbb{R}$ are the coefficients of the series expansions \eqref{eq:3.2}. Then, the function
$
u(x)=\sum_{k=0}^{\infty}\zeta_kx^k
$
is a solution of the differential equation \eqref{cont} if and only if the series
$
u_n=\sum_{k=0}^n\zeta_k \frac{n!}{(n-k)!}
$
is a solution of the difference equation \eqref{fam+}. 
\end{theorem}
As a consequence, we say that eq. \eqref{fam+} represents the discrete counterpart of eq. \eqref{cont} on the lattice $\mathcal{L}^{+}$. The equation \eqref{fam+} can also be seen as the realization of the abstract equation \eqref{discr} on the Rota algebra $(\mathcal{F}, +,\cdot, *_{\Delta^{+}})$. A different realization will be discussed in Section \ref{sec:6}.
Hereafter, for the sake of concreteness, we will list some basic examples of discrete models arising from eq. \eqref{fam+}.

\subsection{A constant coefficient model}

The constant coefficient equation
\beq \label{eq:3.29}
u''+\alpha u'+\beta u=0
\eeq
with  $a_0=\alpha$, $a_l=0,\; l\ge1$, $b_0=\beta$,  $b_l=0,\; l\ge1$
provides us with the difference equation
\beq
u_{n+2}+(\alpha-2)u_{n+1}-(\alpha- \beta-1)u_{n}=0
\eeq
whose independent solutions, inherited from those of eq.\eqref{eq:3.29}  are
\[
u_n^{(1)}=\sum_{k=0}^{n}\frac{(\sqrt{\alpha^2-4\beta}-\alpha)^{k}}{2^k k!} \frac{n!}{(n-k)!}, \quad u_n^{(2)}=\sum_{k=0}^{n}\frac{(-1)^k(\sqrt{\alpha^2-4\beta}+\alpha)^{k}}{2^k k!} \frac{n!}{(n-k)!}.
\]

\subsection{The Airy equation} One of the most interesting examples of the application of the above theory is the Airy equation
\beq \label{Airy}
u''-xu=0
\eeq
which, in quantum mechanics, represents the Schr\"odinger equation for a particle in a one-dimensional constant force field.

The discrete version of it is readily obtained from eq.\eqref{fam+} by taking $a_l=0,\; l\ge0$, $b_0=0$, $b_1=-1$,  $b_l=0,\; l\ge2$. We immediately obtain an interesting discrete model.
\begin{definition}
The difference equation
\beq \label{ADE}
u_{n+2}-2 u_{n+1}+u_{n}-n u_{n-1} =0, \qquad n\in \mathbb{N}
\eeq
will be said to be the discrete Airy equation on $\mathcal{L}^{+}$.
\end{definition}
The following discrete Airy functions provide the two independent solutions of the Airy difference equation \eqref{ADE}:
\[
Ai(n)=\sum_{k=0}^{n}\frac{1}{3^k k!2\cdot 5 \cdots (3k-1)} \frac{n!}{(n-3k)!}, \]
\[ \quad Bi(n)=\sum_{k=0}^{n} \frac{1}{3^{k} k! 4\cdot 7 \cdots (3k+1)} \frac{n!}{(n-(3k+1))!}.
\]

\subsection{The Hermite equation}
The Hermite equation
\beq \label{Hermite}
u''-2xu'+2\lambda u=0
\eeq
can be discretized by taking $a_0=0$, $a_1=-2$,  $a_l=0,\; l\ge2$, $b_0=2\lambda$,   $b_l=0,\; l\ge1$. 
\begin{definition}
The Hermite difference equation is
\beq\label{HDE}
u_{n+2}-2 u_{n+1}+(1+2\lambda -2  n)u_{n}+2n u_{n-1}=0.
\eeq
\end{definition}
By way of an example, we shall consider the Hermite polynomials $H_N(x)$ which are exact solutions of the Hermite equation for $\lambda=N$, being $N$ a nonnegative integer. For instance, for $N=3$, the Hermite polynomial
\beq
H_3(x)=8x^3-12x
\eeq
corresponds to the function
\beq
u_n=8p_3^{+}(n)-12p_1^{+}(n)=  8n^3-24n^2+4n,
\eeq
which is an exact solution of the Hermite difference equation \eqref{HDE} for $\lambda=3$ :
\beq
 u_{n+2}-2u_{n+1}+(7 - 2n )u_n  +2n u_{n-1}=0.
\eeq

\subsection{Continuum limit}

We shall discuss the continuum limits of the large family of discrete models introduced in this section. 
To this aim, we will introduce a step size $h$ and define a regular mesh of points on $\mathbb{R}_{\geq 0}$, indexed by $nh$, $n\in \mathbb{N}$. We will consider the double limit $h\to 0$ and $n\to \infty$ while keeping $nh$ finite. The basic polynomials associated with $\Delta_h= \frac{T^h-1}{h}$ are now
$
p_k^{+}(x)=\prod_{j=0}^{k-1}(x-jh).
$
Thus,
\beq
p_k^{+}(nh)=h^k\prod_{j=0}^{k-1}(n-j)=\frac{n!h^k}{(n-k)!}.
\eeq
Consequently, the coefficients $\zeta_k$ are given by:
\beq
\zeta_k=\frac{1}{h^k}\sum_{j=0}^k\frac{(-1)^{k-j}}{j!(k-j)!}u_j,
\eeq
with 
\beq
u_n=u(nh)=\sum_{k=0}^n\zeta_kp_k^{+}(nh)=\sum_{k=0}^n\frac{n!h^k\zeta_k}{(n-k)!}.
\eeq
The action of $\Delta_h$ on a solution is expressed by
\begin{align}
\Delta_h u =&\frac{1}{h}(u_{n+1}-u_n)= \frac{1}{h}\left(\sum_{k=0}^{n+1}\frac{(n+1)!h^k\zeta_k}{(n+1-k)!}-\sum_{k=0}^n\frac{n!h^k\zeta_k}{(n-k)!}\right) \\ =& \sum_{k=0}^{n+1}h^{k-1}\zeta_k\left(\frac{(n+1)! }{(n+1-k)!}- \frac{n! }{(n-k)!}\right)= \sum_{k=0}^{n+1}k\zeta_k  \frac{ n!h^{k-1}  }{(n-(k-1))!} = \sum_{k=0}^{n+1}k\zeta_k  p_{k-1}^{+}(nh)\notag .
\end{align}
Consequently,
\begin{align}
\Delta_h^2 u =&\frac{1}{h^2}(u_{n+2}-2u_{n+1}+u_n)=  \sum_{k=0}^{n+2}k(k-1)\zeta_k  p_{k-2}^{+}(nh).
\end{align}
The equation \eqref{fam+} on our mesh takes the form
\beq
\frac{1}{h^2}(u_{n+2}-2 u_{n+1}+u_{n})+  \sum_{l=0}^{n}\frac{n!h^l}{(n-l)!}\left[\frac{a_l}{h}(u_{n-l+1}-u_{n-l})+ b_l \, u_{n-l}\right]=0,
\eeq
which reproduces eq. \eqref{cont} in the continuum limit. For instance, the difference Airy equation becomes
\beq
\frac{1}{h^2}(u_{n+2}-2 u_{n+1}+u_{n})-nh  u_{n-1} =0,
\eeq
whereas the difference Hermite equation reads
\beq
\frac{1}{h^2}(u_{n+2}-2 u_{n+1}+u_{n})+2\lambda u_{n}-2 n (u_{n}-u_{n-1})=0.
\eeq
In the continuum limit, they coincide with the differential equations \eqref{Airy} and \eqref{Hermite} respectively.

\section{Frobenius theory for regular singular  points} \label{sec:4}

\subsection{General approach}
In this section, we shall extend the results of the previous analysis to the more general case of differential equations of the form \eqref{cont} possessing regular singular points. To this aim, we first rewrite eq. \eqref{cont} as
\beq \label{eq:4.1}
x^2 u''+ R(x)u'+S(x)u=0,
\eeq
where the functions $R(x) := x^2     \, A(x)$ and $S(x) := x^2 \, B(x)$ are now analytic. In order to restraint the singularity of $A(x)$ to at most a simple pole, we have to impose $R(0)=0$. Thus, we have
\beq \label{eq:4.2}
R(x)=\sum_{k=1}^{\infty}r_kx^k,\quad S(x)=\sum_{k=0}^{\infty}s_kx^k. 
\eeq
According to the classical Frobenius theory, there exists a solution of  eq. \eqref{eq:4.1} of the form
\beq
u(x) = x^\lambda\sum_{k=0}^{\infty}\tilde{\zeta}_k x^k,
\eeq
where $\lambda$ is a root of the indicial polynomial
\beq
\lambda(\lambda-1)+r_1\lambda+s_0.
\eeq
We shall restrict our analysis to the case of real roots, where the highest one is a non-negative integer. Under these hypotheses, the solution can be written as
\beq
u(x) = \sum_{k=0}^{\infty}\tilde{\zeta}_k x^{\lambda+k}
= \sum_{k=\lambda}^{\infty}\tilde{\zeta}_{k-\lambda} x^{k}=\sum_{k=0}^{\infty}\zeta_k x^{k},
\eeq
where the equation will force $\zeta_k=0$ for $k=0,\ldots,\lambda-1$.
Then we proceed as in the regular case.

By analogy with the previous discussion, we introduce the formal series \eqref{eq:3.4},  (assuming $r_0=0$). We formally  define the dynamical system
\beq \label{discr2}
p_2(x)* Q^2 u(x) +V(x)* Q u(x) +W(x)*u(x)=0
\eeq
on the abstract Rota algebra $(\mathcal{F}, +,\cdot, *_{Q})$, where $V(x)= \sum_{l=1}^{\infty}r_lp_l(x)$, $W(x)=\sum_{l=0}^{\infty}s_lp_l(x)$. Explicitly,  eq. \eqref{discr2} has the form
\begin{align}\label{diff02}
&p_2(x)*\left(\sum_{k=0}^{\infty}k(k-1)\zeta_kp_{k-2}(x)\right) +\left(\sum_{l=1}^{\infty}r_lp_l(x)\right)*\left(\sum_{k=0}^{\infty}k\zeta_kp_{k-1}(x)\right)+\notag \\ & \qquad \left(\sum_{l=0}^{\infty}s_lp_l(x)\right)*\left(\sum_{k=0}^{\infty}\zeta_kp_{k}(x)\right)=0,
\end{align}
or, equivalently
\beq \label{eq:4.7}
\sum_{k=0}^{\infty}k(k-1)\zeta_{k}p_{k}(x)+ \sum_{k=0}^{\infty}\sum_{l=0}^{\infty}kr_l\zeta_{k}p_{k+l-1}(x) +\sum_{k=0}^{\infty}\sum_{l=0}^{\infty}s_l\zeta_{k}p_{k+l}(x) =0.
\eeq
By shifting the indices, we have
\beq
\sum_{k=0}^{\infty}\left\{k(k-1)\zeta_{k } +  \sum_{l=0}^{k}[(k-l)r_{l+1} +s_l]\zeta_{k-l }\right\}p_{k}(x)=0.
\eeq
We obtain the recurrence equation
\beq\label{zeta}
[k(k-1)+kr_{1} +s_{0}]\zeta_{k } +  \sum_{l=1}^{k }[(k-l)r_{l+1} +s_l]\zeta_{k-l }=0,\quad k=0, 1,\ldots
\eeq
which implies the relations
\begin{align*}
s_0 \zeta_{0}=&0\\
(r_{1} +s_0)\zeta_{1}+
s_1\zeta_{0}=&0\\
 (2+2r_{1} +s_0)\zeta_{2}+(r_{2} +s_1)\zeta_{1}+s_2\zeta_{0}=&0 \\ \ldots
\end{align*}

\subsection{Difference equation for $u_n$}
The difference equation for $u_n=\sum_{k=0}^{n} \zeta_k p_k(n)$ is obtained from the relation \eqref{eq:4.7}, rewritten as
\beq
\sum_{k=0}^{\infty}\left[k(k-1) p_{k}(n)+ k\sum_{l=0}^{\infty} r_l p_{k+l-1}(n) + \sum_{l=0}^{\infty}s_l p_{k+l}(n)\right]\zeta_k =0.
\eeq
For the case $Q=\Delta^{+}$, by using the inversion formula \eqref{inv} and  the basic polynomials $p_k^{+}(n)$,
we obtain
\beq
\sum_{k=0}^{\infty}\left\{\sum_{j=0}^{k}\left(\frac{n!k(k-1) }{(n-k)!}+  \sum_{l=0}^{\infty}\left[ \frac{n!kr_l}{(n-k-l+1)!}  + \frac{n!s_l }{(n-k-l)!}\right]\right)\frac{(-1)^{k-j}}{j!(k-j)!}u_j \right\}=0.
\eeq
By analogy with the analysis for regular points, we introduce 
\beq
T^s_{kj}=\left(\frac{n!k(k-1) }{(n-k)!}+  \sum_{l=0}^{\infty}\left[ \frac{n!kr_l}{(n-k-l+1)!}  + \frac{n!s_l }{(n-k-l)!}\right]\right)\frac{(-1)^{k-j}}{j!(k-j)!}.
\eeq
Then,
\beq
\sum_{k=0}^{\infty}\sum_{j=0}^{k}T_{kj}u_j=0.
\eeq
By exchanging the indices, we get the linear difference equation for $u_j$
\beq \label{eq:4.18}
\sum_{j=0}^{\infty}\Bigg(\sum_{k=j}^{\infty}T^s_{kj}\Bigg)u_j=\sum_{j=0}^{\infty}C^s_ju_j=0,\quad C^s_j= \sum_{k=j}^{\infty}T^s_{kj}, 
\eeq
where
\beq
C^s_{j}=\sum_{k=j}^{\infty}\left(\frac{n!k(k-1) }{(n-k)!}+  \sum_{l=0}^{\infty}\left[ \frac{n!kr_l}{(n-k-l+1)!}  + \frac{n!s_l }{(n-k-l)!}\right]\right)\frac{(-1)^{k-j}}{j!(k-j)!}.
\eeq
The first addend can be written as
\begin{align}
 \sum_{k=j}^{n}\frac{n!k(k-1) }{(n-k)!}\frac{(-1)^{k-j}}{j!(k-j)!}=&
\binom{n}{j} \sum_{k=j}^{n}(-1)^{k-j}k(k-1)\binom{n-j}{k-j}\notag\\ \nn = &\binom{n}{j} (n(n-1) \delta_{n,j}-2 (n-1) \delta_{n-1,j}+2 \delta_{n-2,j}) \\ = & \,  n(n-1) (\delta_{n,j}-2 \delta_{n-1,j} +2 \delta_{n-2,j}). 
\end{align}
The second and third addend coincide with   eq. \eqref{eq:3.23} and \eqref{eq:3.26}, respectively. The difference equation \eqref{eq:4.18} now becomes
\begin{align}
 \sum_{j=0}^{\infty}C_{j}u_j=&\sum_{j=0}^{\infty}n(n-1)\left( \delta_{n,j}-2 \delta_{n-1,j} +2 \delta_{n-2,j}\right)u_j+  \sum_{l=0}^{\infty} \sum_{j=0}^{\infty}\frac{n!r_l}{(n-l)!}(\delta_{n-l+1,j}-\delta_{n-l,j})u_j\notag \\ \nn & + \sum_{l=0}^{\infty}\sum_{j=0}^{\infty}\frac{n!s_l}{(n-l)!}\delta_{n-l,j}u_j. \end{align}
The latter discussion can be summarised in terms of the following
\begin{theorem} \label{Theorem2}
Consider the class of discrete dynamical systems
\beq \label{discrRS}
n(n-1)(u_{n}-2 u_{n-1}+u_{n-2})+  \sum_{l=0}^{n}\frac{n!}{(n-l)!}\left[r_lu_{n-l+1}-(r_l- s_l)u_{n-l}\right]=0,
\eeq
where $n\in \mathbb{N}$ and $r_l\in \mathbb{R}$, $s_l\in \mathbb{R}$ are the coefficients of the series expansions \eqref{eq:4.2}. Then, the function
$
u(x)=\sum_{k=0}^{\infty}\zeta_kx^k
$
is a solution of the differential equation \eqref{eq:4.2} if and only if the series
$
u_n=\sum_{k=0}^n\zeta_k \frac{n!}{(n-k)!}
$
is a solution of the difference equation \eqref{discrRS}.
\end{theorem}
Thus the family of difference equations \eqref{discrRS} can be regarded as the discrete version of the ODEs \eqref{eq:4.1}. The study of the continuum limit of eq. \eqref{discrRS} is perfectly analogous to the previous case and is left to the reader.

\section{A discrete Bessel equation} \label{sec:5}
The theoretical framework developed above allows us to introduce, for example, a natural discretization for the Bessel equation. Due to the relevance of the Bessel equation in many application contexts, we shall discuss  its discretization in detail here. 

\subsection{The zero-order  Bessel equation}
Let us discuss first the simplest example of the zero-order Bessel differential equation, which reads
\beq
x^2u''+xu'+x^2u=0
\eeq
and admits, as an independent solution, the zero-order Bessel function
 \beq\label{bes0}
 J_{0}(x)=\sum_{k=0}^{\infty}\frac{(-1)^{k}}{(k!)^2 2^{2k}}x^{2k}=1-\frac{x^2}{4}+\frac{x^4}{64}-\frac{x^6}{2304} +\cdots
 \eeq
In this case, $R(x)=x$, $S(x)=x^2$. For $n\ge2$, we obtain a remarkable difference equation.
\begin{definition}
The  zero-order discrete Bessel equation is given by
\beq\label{bes0disc}
nu_n- (2n-1)u_{n-1}+2(n-1) u_{n-2}=0.
\eeq
\end{definition}
The discrete version of the Bessel function \eqref{bes0} $J_0(x)$ arising from our theory is
\beq
u_n=\sum_{k=0}^{\infty}\frac{(-1)^{k}}{(k!)^2 2^{2k}}p_{2k}^{+}(n)=\sum_{k=0}^{[n/2]}\frac{(-1)^{k}n!}{2^{2k}(n-2k)!(k!)^2},
\eeq
which represents a solution of the difference equation \eqref{bes0disc}.
Precisely, from eq. \eqref{bes0disc} we obtain the relation
\begin{align}
&n\sum_{k=0}^{\infty}\frac{(-1)^{k}n!}{2^{2k}(n-2k)!(k!)^2 }-(2n-1)\sum_{k=0}^{\infty}\frac{(-1)^{k}(n-1)!}{2^{2k}(n-1-2k)!(k!)^2 }\notag\\  
&+2(n-1) \sum_{k=0}^{\infty}\frac{(-1)^{k}(n-2)!}{2^{2k}(n-2-2k)!(k!)^2 }. 
\end{align}
It can be simplified, leading us to the expression
\beq
\sum_{k=0}^{\infty}\frac{(-1)^k}{2^{2k}n(k!)^2(n-2k)!} [ n^2-(4 k+1) (n-2 k)],
\eeq
which vanishes for $n\ge 2$.

\subsubsection{Continuum limit}
With the same construction as in the case of  regular  points, one can easily derive, in a lattice of step $h$, the expression
\beq
(n+2)^2  u_{n+2}-  2(n+2)^2u_{n+1}+ (n+2)^2u_{n}+( n +2) u_{n+1}-(n+2) u_{n} +  (n+1)(n+2)h^2  u_{n}=0,
\eeq
which can be written as:
\begin{align} 
((n+2) h)^2\frac{ u_{n+2}-  2 u_{n+1}+  u_{n}}{h^2}+( n +2)h\frac{ u_{n+1}- u_{n}}{h} +   (n+1)(n+2)h^2 u_{n}=0.
\end{align}
In the limit ($h\to 0$, $ n\to \infty$ and $nh$ finite) we easily recover the zero-order Bessel equation.

\subsection{The Bessel equation of order $\nu$}

Let us extend the previous discussion to the Bessel equation of order $\nu\in\mathbb{N}\backslash\{0\}$:
\beq\label{besseldif}
x^2u''+xu'+(x^2-\nu^2)u=0.
\eeq
In this case, $R(x)=x$, $S(x)=x^2-\nu^2$. Thus, the non-vanishing coefficients are
\beq
r_1=1,\quad s_0=-\nu^2,\quad s_2=1.
\eeq
Therefore, eq. \eqref{discrRS} converts into a simpler relation.
\begin{definition}
The difference equation
\beq\label{besselnu}
(n^2-\nu^2)u_{n}
-n(2n-1)u_{n-1}
+2n(n-1)u_{n-2} 
=0
\eeq
is said to be the discrete Bessel equation of order  $\nu\in\mathbb{N}\backslash\{0\}$. 
\end{definition}

\subsubsection{Exact solutions}

From the Bessel function  of the first class and order $\nu$
\beq
J_{\nu}(x)=\sum_{k=0}^{\infty}\frac{(-1)^k}{2^{2k+\nu}k!(k+\nu)!}x^{2k+\nu}
\eeq
we can define its discrete counterpart
\beq
u_{n} =\sum_{k=0}^{[(n-\nu)/2]}\frac{(-1)^k}{2^{2k+\nu}k!(k+\nu)!}p_{2k+\nu}^{+}(n),
\eeq
which is an exact solution of \eqref{besselnu}.

\section{Discretization on the real negative axis} \label{sec:6}

The formal approach presented in this work can be easily adapted in order to produce different discretization schemes, related to different choices of delta operators and of the corresponding lattices (i.e., of Rota algebras).

To illustrate this crucial aspect, in this section, we shall briefly discuss the case of $Q=\Delta^{-}=\mathbf{1}-T^{-1}$. The associated basic polynomials $p_n^{-}(x)$ can be  written explicitly as
 \beq
p_0^{-}(x)=1,\quad p_k^{-}(x)=\prod_{j=0}^{k-1}(x+j), \qquad k=0,1,2,\ldots. \eeq
They satisfy the identity
 \[
 p_k^{-}(-x)=\prod_{j=0}^{k-1}(-x+j)=(-1)^k\prod_{j=0}^{k-1}(x-j)=(-1)^k p_k^{+}(x).
 \]
The zeros of these polynomials define a new, uniform lattice, denoted as $\mathcal{L}^{-}$. We shall index its points as $x=-m$, $m=0,1,2,\ldots$. The formal expansion \eqref{eq:3.4} on $\mathcal{L}^{-}$ has the form
\beq
u_{-m}\equiv u(-m)=\sum_{k=0}^{\infty} \zeta_k\, p_k^{-}(-m),
\eeq
which converts into the finite sum
\beq
u_{-m}=\sum_{k=0}^{m} \zeta_k\, p_k^{-}(-m)=\sum_{k=0}^{m}  \frac{(-1)^km!}{(m-k)!}\zeta_k, \qquad m=0,1,2,\ldots
\eeq
The coefficients $\zeta_k$ can be inverted, giving us the expression
\beq
\zeta_k=\sum_{j=0}^{k} \frac{(-1)^{j}}{j!(k-j)!}u_{-j},\quad k=0,1,2,\ldots
\eeq

The Frobenius analysis for this discretization scheme can be performed in complete analogy to the discussion of the previous sections. We report here the main results of this analysis, leaving the technical details to the reader. We distinguish two cases.

\vspace{2mm}

(A) For the regular case, the discretization of eq. \eqref{cont} on $\mathcal{L}^{-}$ provides us with a family of equations analogous to that of eq. \eqref{fam+}:
\begin{definition}
The discrete version of eq. \eqref{cont} on the real negative axis reads
\beq \label{eq:6.5}
u_{-m-2}-2 u_{-m-1}+u_{-m}+\sum_{l=0}^{m}\frac{m!}{(m-l)!}\Big(-a_l u_{-m+l-1}+(a_l+b_l)u_{-m+l}\Big)=0,
\eeq
where $m\in \mathbb{N}$ and $a_l$, $b_l\in \mathbb{R}$ are the coefficients of the series expansions \eqref{eq:3.2}.
\end{definition}

\vspace{2mm}

(B) For the regular singular case, following the same philosophy, the same discretization scheme applied to eq. \eqref{eq:4.1} gives us another interesting dynamical model, which is the companion of the family \eqref{discrRS}, defined on $\mathcal{L}^{-}$.
\begin{definition}
The discrete version of eq. \eqref{eq:4.1} on the real negative axis is the family of models
\beq \label{eq:6.6}
m(m-1)(u_{-m}-2 u_{-m+1}+ u_{-m+2})+\sum_{l=0}^{m}(-1)^{l} \frac{m!}{(m-l)!}\Big(-r_l u_{-m+l-1}+ (r_l+s_l)u_{-m+l}\Big)=0, 
\eeq
where $m\in \mathbb{N}$ and $r_l\in \mathbb{R}$, $s_l\in \mathbb{R}$ are the coefficients of the series expansions \eqref{eq:4.2}.
\end{definition}
By complete analogy with Theorems \ref{Theorem1} and \ref{Theorem2}, we can state a general result for the discretizations on the real negative axis.
\begin{theorem} \label{Theorem3}
The power series $u_{-m}=\sum_{k=0}^{m}  \zeta_k \frac{(-1)^km!}{(m-k)!}$  is an exact solution of eq. \eqref{eq:6.5} (resp. \eqref{eq:6.6})  if and only if the function $u(x)=\sum_{k=0}^{\infty}\zeta_k x^{k}$ is an analytic solution of eq. \eqref{cont} (resp. \eqref{eq:4.2}). 
\end{theorem}
The proof of this statement parallels that of Theorems \ref{Theorem1} and \ref{Theorem2} and is not given here.

For the sake of clarity, we shall write explicitly the discrete versions on $\mathcal{L}^{-}$ of the paradigmatic cases of the  Hermite and Bessel equations.

\begin{definition}
We shall define the discrete Hermite equation on $\mathcal{L}^{-}$ as:
\beq
(1-2m +2 \lambda)u_{-m}-2 u_{-m-1}+u_{-m-2}+2 m u_{-m+1}=0, \qquad m\in \mathbb{N}.
\eeq
\end{definition}
\begin{definition}
The discrete Bessel equation of order zero and order $\nu\in\mathbb{N}\backslash\{0\}$ over $\mathcal{L}^{-}$ are, respectively
\beq \label{eq:6.8}
m u_{-m}-(2m-1)u_{-m+1}+2(m-1) u_{-m+2}=0, \qquad m \geq 2
\eeq
and
\beq \label{eq:6.9}
(m^2-\nu^2)u_{-m}-m(2m-1)u_{-m+1}+2m(m-1)u_{-m+2}=0.
\eeq
\end{definition}
\begin{remark}
Although eq. \eqref{eq:6.6} is defined for $m\geq 0$, the discrete Bessel equation of order zero \eqref{eq:6.8} is defined for $m\geq 2$ due to the suppression of the factor $m(m-1)$ in its simplified form. The same argument applies to eq. \eqref{bes0disc}.
\end{remark}

\section{Future perspectives} \label{sec:7}

This work is part of a research project aimed at investigating new discretization schemes for ODEs and PDEs. The main result of this article is a novel discretization of second-order differential equations with analytic coefficients, which is adapted to their singularity properties, according to the Frobenius theorem. The discretization schemes we propose are effective and lead to classes of difference equations which admit exact solutions inherited from their continuous counterparts in a simple and natural way.

We plan to consider the problem of discretizing of a class of relevant quantum-mechanical models and quantum field theories from the point of view of the theoretical framework developed both in this work and in \cite{PTJDE, PTprep}, based on the finite operator theory and in particular on Rota algebras. The general study of initial and boundary value problems in our theoretical framework is a relevant open problem. The analysis of the symmetry properties of our discrete models by means of Lie group theory is another important goal for future research.   
Work in this direction is underway.

\subsection*{Acknowledgements}
We wish to dedicate this article to the loving memory of Professor Pavel Winternitz, who participated in the early stages of this research project.

D. R. acknowledges the financial support of EXINA S.L.; M.A.R. wishes to thank the Grupo UCM - F\'isica Matem\'atica for financial support.

The research of P. T. has been supported by the Severo Ochoa Programme for Centres of Excellence in R\&D
(CEX2019-000904-S), Ministerio de Ciencia, Innovaci\'{o}n y Universidades y Agencia Estatal de Investigaci\'on, Spain.  P.T. is a member of the Gruppo Nazionale di Fisica Matematica (GNFM) of the Istituto Nazionale di Alta Matematica (INdAM).

\end{document}